\documentclass[a4paper,twocolumn,11pt,unpublished]{quantumarticle}
\pdfoutput=1
\usepackage[utf8]{inputenc}
\usepackage[english]{babel}
\usepackage[T1]{fontenc}
\usepackage{amsmath,amssymb,amsthm}
\usepackage{graphicx}
\usepackage{dcolumn}
\usepackage{bm}
\usepackage[numbers]{natbib}
\usepackage{hyperref}

\newtheorem{proposition}{Proposition}
\newtheorem{theorem}{Theorem}
\newtheorem{lemma}{Lemma}
\newtheorem{definition}{Definition}
\newtheorem{remark}{Remark}
\newtheorem{protocol}{Protocol}

\begin{document}

\title{Certifying Nonclassical Proper-Time Histories with a Quantum Clock}

\author{Shuai Zeng}
\email{zengshuai@cqupt.edu.cn}
\affiliation{School of Communication and Information Engineering, Chongqing University of Posts and Telecommunications, Chongqing 400065, P.R. China}

\begin{abstract}
  Quantum clocks can acquire relativistic phases from motional or gravitational proper-time differences, but reduced clock dephasing alone does not certify nonclassical proper-time histories. We formulate this distinction as a channel-certification problem. First, we show that any two-level single-time dephasing signal, including one generated by an effective quantum proper-time label, admits a classical random proper-time representation. We then define the convex set of classical mixtures of experimentally specified proper-time histories and prove a Choi-rank separation criterion for conditioned coherent history recombination. A two-branch Ramsey protocol gives explicit bright- and dark-port population witnesses outside this classical set. The certification is operational and relative to the specified history set: it rules out classical mixtures of the same implemented proper-time histories, not arbitrary classical protocols with different histories or controls.
\end{abstract}

\maketitle

\section{Introduction}
\label{sec:introduction}

Precision clocks now probe regimes where relativity and quantum mechanics meet. In trapped-ion optical clocks, motional quantum states can imprint relativistic time-dilation phases on clock evolution, leading to vacuum-induced, squeezing-enhanced, and quantum corrections to the clock frequency, as well as clock--motion correlations~\cite{Sorci2026,SmithAhmadi2020,Zych2011,Pikovski2015}. These developments raise a sharper question than whether proper time can affect a quantum clock: \emph{when does an observed clock signal certify nonclassical proper-time histories rather than a classical random proper-time parameter?} This operational question is complementary to relational-time and quantum-reference-frame approaches, where clocks or reference frames are themselves treated as quantum systems~\cite{PageWootters1983,Giacomini2019}.

This certification question is important because several quantum-looking signatures can still have classical random-time explanations. A clock phase shift, a frequency shift, or a loss of Ramsey visibility may arise from quantum motion, but once the motion is traced out, the observed reduced clock signal may be statistically indistinguishable from a classical distribution of proper times. We separate quantum-generated clock noise from certified nonclassical proper-time histories.

This distinction matters for forthcoming clock experiments: relativistic phases, clock--motion correlations, and visibility loss can be observed without yet excluding classical stochastic descriptions of the clock's proper-time history. A certification criterion must identify when the measured operation requires coherent recombination of proper-time histories.

Our result is a three-level certification hierarchy for proper-time signatures. At the first level, single-time reduced clock dephasing, even when generated by an effective quantum proper-time label, can be represented by a classical random proper-time distribution (Sec.~\ref{sec:classical_random}). At the second level, averaged controlled histories form classical mixtures of a specified history set generated by proper-time evolutions and known clock controls (Sec.~\ref{sec:proper_time_histories}). At the third level, conditioned Ramsey histories generated by coherent recombination of the same branches can leave this classical history-mixture set (Secs.~\ref{sec:separation}--\ref{sec:minimal_witness}). This hierarchy is the central object of the paper: it specifies what kind of clock evidence remains classically simulable and what kind certifies nonclassical proper-time histories.

The resulting distinction is operational. Quantum erasure and clock interferometry are established ingredients~\cite{ScullyDruhl1982,Margalit2015}; the new element is the simulability boundary between averaged proper-time histories and conditioned coherent history recombination. The same specified histories are classically simulable when averaged, but can leave the classical set when coherently recombined and conditioned on a selective erasure of which-proper-time-history information. This separation is witnessed by a minimal two-branch Ramsey protocol and a simple postselected population imbalance.

\paragraph{Contributions and assumptions.}
This paper makes four contributions. First, it proves a dephasing no-go result (Propositions~\ref{prop:dephasing}--\ref{prop:effective}): single-time reduced clock dephasing, even when generated by a quantum motional label, is simulable by a classical random proper-time distribution. Second, it defines the convex set $\mathrm{CPTH}(\mathcal H)$ of classical mixtures of an experimentally specified set $\mathcal H$ of proper-time histories (Definitions~\ref{def:history}--\ref{def:CPTH}). Third, it proves a Choi-rank separation criterion (Theorem~\ref{thm:choi}) showing that conditioned coherent recombination can produce a channel outside this set. Fourth, it constructs a minimal two-branch Ramsey witness (Protocol~\ref{prot:ramsey}) with a high-probability bright port and a low-probability dark port. The certification is relative to the specified history set $\mathcal H$; it does not claim to exclude arbitrary classical protocols with different histories or additional controls.

The key assumption is that the measurement apparatus implements a history-erasing projection in a basis that coherently recombines a finite, experimentally specified set of proper-time histories. The witness certifies failure of classical mixtures of \emph{these} histories; it does not provide a necessary-and-sufficient characterization of all forms of quantum proper-time nonclassicality.

\section{Classical random proper time and reduced dephasing}
\label{sec:classical_random}

Let the clock Hamiltonian be
\begin{equation}
  H_C=\sum_n E_n |n\rangle\langle n|.
\end{equation}
For a classical proper time $\tau$, the clock evolves as $U_\tau=e^{-iH_C\tau}$. A classical random proper-time channel is a convex mixture of such evolutions,
\begin{equation}
  \mathcal E_{\rm cl}(\rho)
  =
  \int d\mu(\tau)\,
  U_\tau \rho U_\tau^\dagger,
  \label{eq:Ecl}
\end{equation}
where $d\mu(\tau)$ is a probability measure. In this model, each run has a definite classical proper time, while the value fluctuates from shot to shot.

\begin{proposition}[Two-level dephasing no-go]
  \label{prop:dephasing}
  For a two-level clock with transition frequency $\omega$, every single-time dephasing signal
  \begin{equation*}
    \rho_{ge}\mapsto \Gamma\rho_{ge},
    \qquad |\Gamma|\le 1,
  \end{equation*}
  admits a classical random proper-time representation
  \begin{equation*}
    \Gamma = \int d\mu(\tau)\,e^{-i\omega\tau},
  \end{equation*}
  because the convex hull of phases $e^{-i\omega\tau}$ is the unit disk.
\end{proposition}

\begin{proof}
  Write $\Gamma=re^{i\phi}$ with $0\le r\le1$. Choose $\tau_1,\tau_2$ such that
  $e^{-i\omega\tau_1}=e^{i\phi}$ and $e^{-i\omega\tau_2}=-e^{i\phi}$.
  The probability measure
  $d\mu(\tau)=\frac{1+r}{2}\delta(\tau-\tau_1)d\tau+\frac{1-r}{2}\delta(\tau-\tau_2)d\tau$
  gives $\Gamma$. A single Ramsey visibility or phase shift therefore cannot certify nonclassical proper time.
\end{proof}

\begin{proposition}[Effective proper-time label]
  \label{prop:effective}
  Consider the effective evolution $U=e^{-iH_C\otimes\hat\tau}$, where $\hat\tau$ is a self-adjoint label of motional proper-time branches. For an initially uncorrelated state $\rho_C\otimes\rho_M$, tracing out the motional degree of freedom gives reduced clock coherences that are characteristic functions of the classical probability measure $d\mu_{\rho_M}(\tau)=\operatorname{Tr}_M[\rho_M P(d\tau)]$, where $P(d\tau)$ is the spectral measure of $\hat\tau$.
\end{proposition}

\begin{proof}
  Under the effective model $U=e^{-iH_C\otimes\hat\tau}$, where $\hat\tau$ is a self-adjoint label of motional proper-time branches (not a universal time operator conjugate to the Hamiltonian), an initially uncorrelated state $\rho_C\otimes\rho_M$ evolves and, after tracing out motion, yields
  \begin{equation}
    \Gamma_{mn}
    =
    \operatorname{Tr}_M
    \left[
      \rho_M e^{-i(E_m-E_n)\hat\tau}
    \right].
    \label{eq:Gamma_mn}
  \end{equation}
  Let $P(d\tau)$ be the spectral measure of $\hat\tau$. Then
  \begin{equation}
    \Gamma_{mn}
    =
    \int e^{-i(E_m-E_n)\tau}\,
    \operatorname{Tr}_M[\rho_M P(d\tau)].
    \label{eq:Gamma_spectral}
  \end{equation}
  The measure $d\mu_{\rho_M}(\tau)=\operatorname{Tr}_M[\rho_M P(d\tau)]$ is a probability measure. Clock-only data can be quantum in origin but classical in reduced statistics.
\end{proof}

A concrete branch-label model for Fock-state motional branches is given in Appendix~\ref{app:effective_label}. In summary, reduced dephasing alone cannot certify nonclassical proper-time histories --- a genuine witness must access more than a single-time reduced clock channel.

\section{Classical mixtures of specified proper-time histories}
\label{sec:proper_time_histories}

\begin{definition}[Proper-time history]
  \label{def:history}
  A proper-time history is a sequence
  \begin{equation}
    h=(\tau_1,\ldots,\tau_L)
  \end{equation}
  interleaved with known clock controls $R_1,\ldots,R_{L-1}$. The corresponding clock unitary is
  \begin{equation}
    V_h
    =
    U_{\tau_L}R_{L-1}U_{\tau_{L-1}}\cdots R_1 U_{\tau_1}.
  \end{equation}
\end{definition}

For a given experiment, let $\mathcal H$ denote the specified set of histories that may be classically selected. Each history $h\in\mathcal H$ induces a channel $\mathcal U_h(\rho)=V_h\rho V_h^\dagger$.

\begin{definition}[Classical history-instrument cone]
  \label{def:CPTH_inst}
  The positive cone generated by classical proper-time-history instrument outcomes is
  \begin{equation}
    \mathsf{CPTH}_{\rm inst}(\mathcal H)
    =
    \left\{
      \sum_{h\in\mathcal H} q_h V_h(\cdot)V_h^\dagger:
      q_h\ge0
    \right\}.
    \label{eq:CPTH_inst}
  \end{equation}
  A physical subnormalized outcome in this cone additionally satisfies $\sum_{h\in\mathcal H}q_h\le1$.
\end{definition}

\begin{definition}[Normalized classical free set]
  \label{def:CPTH}
  For normalized postselected channels, the associated free set is the convex hull
  \begin{equation}
    \mathsf{CPTH}(\mathcal H)
    =
    {\rm conv}\left\{
      \mathcal U_h:\rho\mapsto V_h\rho V_h^\dagger,
      \ h\in\mathcal H
    \right\}.
    \label{eq:CPTHH}
  \end{equation}
\end{definition}

This set accommodates classical uncertainty over histories, classical postselection, outcome-dependent reweighting, and arbitrary known clock controls within the specified histories. It excludes coherent recombination between distinct histories --- terms of the form $V_h\rho V_{h'}^\dagger$ with $h\ne h'$. A violation of $\mathsf{CPTH}(\mathcal H)$ should therefore be read as a failure of classical mixtures of the same experimentally specified proper-time histories, not as a claim about every conceivable classical history one could define with different controls.

\begin{lemma}[Choi membership feasibility]
  \label{lem:choi_feasibility}
  In the Choi representation, membership in the normalized free set is the finite convex feasibility condition
  \begin{equation}
    J(\Phi)
    =
    \sum_{h\in\mathcal H}p_h J(\mathcal U_h),
    \qquad
    p_h\ge0,
    \qquad
    \sum_{h\in\mathcal H}p_h=1,
    \label{eq:choi_normalized}
  \end{equation}
  where $J$ denotes the Choi map. For the subnormalized instrument cone, the corresponding condition is
  \begin{equation}
    J(\mathcal I)
    =
    \sum_{h\in\mathcal H}q_h J(\mathcal U_h),
    \qquad
    q_h\ge0.
    \label{eq:choi_subnormalized}
  \end{equation}
\end{lemma}

The dependence on $\mathcal H$ is essential: a violation of $\mathsf{CPTH}(\mathcal H)$ refers to classical mixtures of the same experimentally specified histories.

\section{Choi-rank separation theorem}
\label{sec:separation}

Suppose a quantum branch register coherently labels the experimentally implemented proper-time histories in $\mathcal H$ and is then measured in a basis that erases which-history information. The conditioned clock operation can have a Kraus operator
\begin{equation}
  K_m
  =
  \sum_{h\in\mathcal H} c_{m,h}V_h,
  \label{eq:Kraus}
\end{equation}
so that $\mathcal I_m(\rho)=K_m\rho K_m^\dagger$. Expanding gives coherent cross terms $V_h\rho V_{h'}^\dagger$ with $h\ne h'$, which are absent from $\mathsf{CPTH}_{\rm inst}(\mathcal H)$.

\begin{theorem}[Choi-rank separation for conditioned proper-time-history recombination]
  \label{thm:choi}
  Assume a conditioned operation has a single Kraus operator $K_m=\sum_{h\in\mathcal H} c_{m,h}V_h$ and satisfies
  \begin{equation}
    K_m^\dagger K_m = p_m I,
    \qquad p_m>0.
    \label{eq:subunitary}
  \end{equation}
  Then $\mathcal I_m(\rho)=p_m W_m\rho W_m^\dagger$ with $W_m=K_m/\sqrt{p_m}$ is a subnormalized unitary channel. If
  \begin{equation}
    W_m\not\propto V_h
    \qquad \text{for all } h\in\mathcal H,
    \label{eq:notprop}
  \end{equation}
  then $\mathcal I_m\notin\mathsf{CPTH}_{\rm inst}(\mathcal H)$ and the normalized channel $\rho\mapsto W_m\rho W_m^\dagger$ is not in $\mathsf{CPTH}(\mathcal H)$.
\end{theorem}

\begin{proof}
  A unitary channel has a rank-one Choi operator $p_m|W_m\rangle\rangle\langle\langle W_m|$, whereas a nonnegative mixture $\sum_{h\in\mathcal H} q_h |V_h\rangle\rangle\langle\langle V_h|$ can be rank one only if all nonzero $|V_h\rangle\rangle$ are parallel. Hence a classical mixture can equal $W_m(\cdot)W_m^\dagger$ only if every nonzero component is proportional to $W_m$. Condition~\eqref{eq:notprop} rules this out. The full Choi argument is given in Appendix~\ref{app:choi}.
\end{proof}

\begin{remark}[Scope and limitations]
  \label{rem:scope}
  The theorem covers single-Kraus, subnormalized-unitary outcomes. For $K^\dagger K\not\propto I$ or multi-Kraus outcomes, one must directly apply the Choi cone membership test of Lemma~\ref{lem:choi_feasibility}. The minimal Ramsey ports constructed below satisfy the subnormalized-unitary condition. The relevant resource is the conditioned coherent recombination of a specified set of proper-time histories into an operation that no classical mixture of those histories can reproduce. A full characterization of all nonclassical history processes lies beyond the minimal witness developed here.
\end{remark}

\section{Minimal Ramsey witness}
\label{sec:minimal_witness}

Consider two branches
\begin{equation}
  U_\pm=e^{\mp i\theta Z/2}
  \label{eq:branches}
\end{equation}
with an intermediate Ramsey $\pi/2$ pulse $R$. The two-history set is
\begin{equation}
  \mathcal H_2=\{+,-\},
  \qquad
  V_+=U_+R\,U_+,
  \quad
  V_-=U_-R\,U_-,
  \label{eq:Vpm}
\end{equation}
where
\begin{equation}
  R=\frac{1}{\sqrt2}\!\begin{pmatrix}1&1\\1&-1\end{pmatrix}.
\end{equation}
If the motional branch is ignored, the clock sees the averaged channel
\begin{equation}
  \mathcal E_{\rm avg}(\rho)
  =
  \frac12 V_+\rho V_+^\dagger
  +
  \frac12 V_-\rho V_-^\dagger,
  \label{eq:Eavg}
\end{equation}
which is an element of $\mathsf{CPTH}(\mathcal H_2)$.

If instead the motion is measured in a history-erasing basis, the conditioned Kraus operator for the symmetric outcome is
\begin{equation}
  K_+=\frac12(V_+ + V_-).
  \label{eq:Kplus}
\end{equation}
A direct calculation (Appendix~\ref{app:ramsey}) gives
\begin{equation}
  \begin{aligned}
    K_+
    &=
    \frac{1}{\sqrt2}
    \begin{pmatrix}
      \cos\theta & 1\\
      1 & -\cos\theta
    \end{pmatrix},\\[4pt]
    K_+^\dagger K_+
    &=
    \frac{1+\cos^2\theta}{2}I.
  \end{aligned}
  \label{eq:Kplus_explicit}
\end{equation}
The conditioned operation is therefore a subnormalized unitary channel. For $\theta\notin\pi\mathbb Z$, the corresponding unitary $W_+$ is not proportional to $V_+$ or $V_-$, and Theorem~\ref{thm:choi} places the conditioned channel outside the classical mixture of the same two histories.

\paragraph{Witness table.}
The two outcomes of the Ramsey protocol are summarized in Table~\ref{tab:witness}. The bright-port outcome ($+$) occurs with high probability and carries a small signal; the dark-port outcome ($-$) occurs with low probability but carries unit conditional contrast.

\begin{table*}[t]
  \centering
  \caption{Two-history Ramsey witness. $K_\pm$ are the Kraus operators, $p_{\rm succ}$ is the success probability, and the last column gives the conditional clock population imbalance for input $|g\rangle$.}
  \label{tab:witness}
  \begin{tabular}{lccc}
    \hline
    Outcome & $K$ & $p_{\rm succ}$ & $\langle Z\rangle$ for $|g\rangle$ \\
    \hline
    Bright ($+$) & $K_+ = (V_++V_-)/2$ & $(1+\cos^2\theta)/2$ & $\displaystyle-\frac{\sin^2\theta}{1+\cos^2\theta}$\\[8pt]
    Dark ($-$) & $K_- = (V_+-V_-)/2$ & $\sin^2\theta/2$ & $1$\\
    \hline
  \end{tabular}
\end{table*}

For small $\theta$, the bright-port signal scales as $\theta^2/2$, requiring $N_{\rm trial}^{(+)}\sim4s^2/\theta^4$ experimental trials for $s\sigma$ significance. The dark port has unit conditional contrast and success probability $p_-\simeq\theta^2/2$, giving $N_{\rm trial}^{(-)}\sim2s^2/\theta^2$.

For the bright port with input $|g\rangle$, every channel in $\mathsf{CPTH}(\mathcal H_2)$ gives $\langle Z\rangle_{\rm cl}=0$, while
\begin{equation}
  \langle Z\rangle_+
  =
  \frac{\cos^2\theta-1}{\cos^2\theta+1}
  \simeq -\frac{\theta^2}{2}
  \qquad (\theta\ll1).
  \label{eq:Zplus}
\end{equation}
A nonzero postselected imbalance thus certifies the failure of classical mixtures of the specified Ramsey histories.

The dark-port outcome is
\begin{equation}
  \begin{aligned}
    K_-&=\frac12(V_+-V_-)=-\frac{i\sin\theta}{\sqrt2}I,\\
    p_-&=\frac{\sin^2\theta}{2},\\
    \langle Z\rangle_-&=1
    \qquad (|g\rangle\ \text{input}).
  \end{aligned}
  \label{eq:darkport}
\end{equation}
It delivers unit conditional contrast with a low postselection probability.

\begin{protocol}[Two-history Ramsey witness]
  \label{prot:ramsey}
  \begin{enumerate}
    \item Prepare a quantum branch register in a coherent superposition labeling the specified proper-time histories.
    \item Let the clock undergo the Ramsey histories $V_h$ controlled by the branch register.
    \item Measure the branch register in a history-erasing basis.
    \item Measure the clock $Z$ observable on the postselected ensemble.
    \item Compare the conditional signal $\langle Z\rangle$ against the classical-mixture prediction $\langle Z\rangle_{\rm cl}=0$ for $\mathsf{CPTH}(\mathcal H_2)$.
  \end{enumerate}
\end{protocol}

\begin{figure*}[t]
  \centering
  \begin{minipage}[t]{0.32\textwidth}
    \centering
    \includegraphics[width=\textwidth]{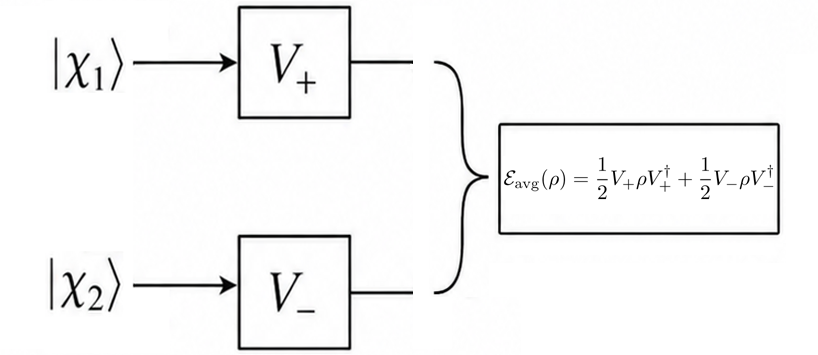}
    \vspace{-4pt}
    \par\footnotesize\sffamily(a)
  \end{minipage}
  \hfill
  \begin{minipage}[t]{0.32\textwidth}
    \centering
    \includegraphics[width=\textwidth]{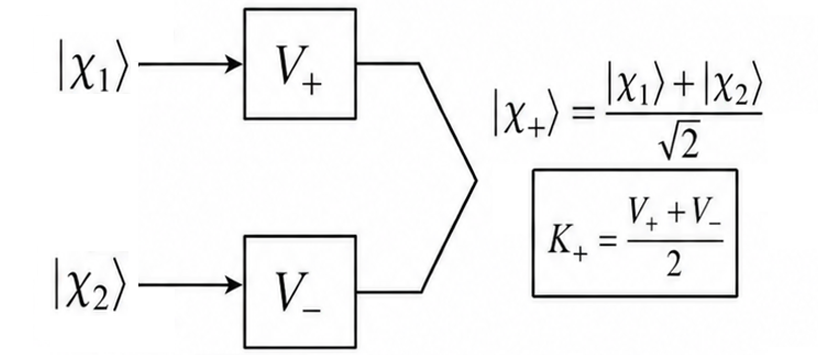}
    \vspace{-4pt}
    \par\footnotesize\sffamily(b)
  \end{minipage}
  \hfill
  \begin{minipage}[t]{0.32\textwidth}
    \centering
    \includegraphics[width=\textwidth]{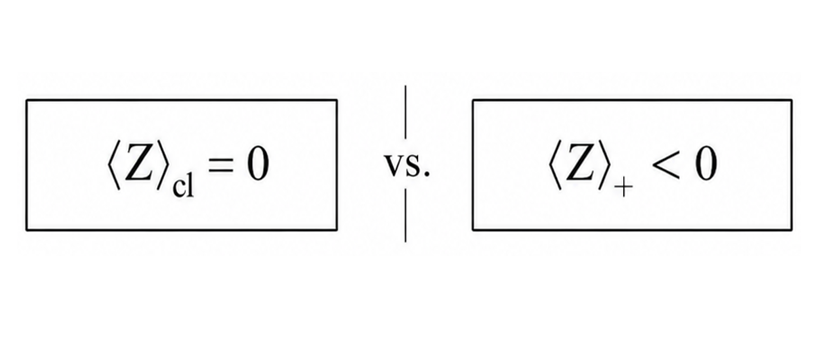}
    \vspace{-4pt}
    \par\footnotesize\sffamily(c)
  \end{minipage}
  \caption{\label{fig:histories}%
    \textbf{Classical mixtures and selective erasure of proper-time histories.}
    (a) Two motional branches $|\chi_1\rangle,|\chi_2\rangle$ label the specified proper-time histories $V_\pm=U_\pm R\,U_\pm$, with $U_\pm=e^{\mp i\theta Z/2}$. If the motional label is ignored, the clock sees the free mixture $\mathcal E_{\rm avg}\in\mathsf{CPTH}(\mathcal H_2)$.
    (b) Measuring the motion in the history-erasing basis produces the bright port $K_+=(V_++V_-)/2$, while the orthogonal erasing outcome gives the dark port $K_-=(V_+-V_-)/2$. The witness depends on erasing this proper-time-history label.
    (c) For input $|g\rangle$, every free mixture of the specified histories predicts $\langle Z\rangle_{\rm cl}=0$, while the bright port yields the negative imbalance in Eq.~\eqref{eq:Zplus}; the complementary dark-port signal is summarized in Table~\ref{tab:witness}.
  }
\end{figure*}

\section{Selective erasure and null conditions}
\label{sec:selective_erasure}

The witness turns on a selective condition: the erased degree of freedom must be the label of the proper-time history itself, not merely an ancillary path or motional tag. We now separate proper-time-history erasure from ordinary quantum erasure.

\subsection{Proper-time-history label condition}

A motional label is a proper-time-history label only if
\begin{equation}
  |\chi_h\rangle \longleftrightarrow V_h
  \label{eq:label_condition}
\end{equation}
with distinct history unitaries generated by different effective proper-time evolutions. In a trapped-ion context, a sufficient physical condition is
\begin{equation}
  \Delta\langle p^2\rangle\neq0,
  \qquad
  \Delta\tau\simeq -\frac{t}{2m^2c^2}\Delta\langle p^2\rangle\neq0.
  \label{eq:sufficient_physical}
\end{equation}
The motional system thereby labels the clock histories themselves. If the motional branches induce the same history unitary, $V_1=V_2$, the setting remains inside the same classical history-mixture description and no certification occurs.

\subsection{Ordinary erasure versus proper-time-history erasure}

Ordinary quantum erasure~\cite{ScullyDruhl1982} restores interference between path or motional labels. It is an established ingredient of quantum optics and matter-wave interferometry. The present protocol uses erasure for a different purpose: certification that the conditioned clock operation cannot be simulated by classical mixtures of the specified proper-time histories.

Three null conditions make the distinction concrete:
\begin{enumerate}
  \item \emph{Identical history unitaries.} If $V_1=V_2$, erasure produces no cross terms between distinct histories and the conditioned operation remains in $\mathsf{CPTH}(\mathcal H)$.
  \item \emph{Removal of intermediate Ramsey control.} Removing the intermediate $R$ pulse reduces the protocol to a simpler phase-erasure setting and removes the population-witness structure used here. The present witness is designed to certify history recombination through a nontrivial clock control sequence, rather than through a single diagonal phase-erasure event. The certification itself is determined by membership or non-membership in $\mathsf{CPTH}(\mathcal H)$, not by a commutator condition between $V_+$ and $V_-$.
  \item \emph{Unrelated label.} If the erased degree of freedom is unrelated to the induced proper-time histories, the erasure event is an ordinary quantum-eraser test rather than a proper-time-history witness.
\end{enumerate}

Certification therefore demands more than a successful erasure event: the erased label must encode the proper-time-history identity, and the conditioned operation must satisfy $\mathcal I_m\notin\mathsf{CPTH}_{\rm inst}(\mathcal H)$ for the specified history set.

\section{Physical realization and target scales}
\label{sec:physical}

The protocol is mapped to a trapped-ion optical-clock setting. This section identifies the physical ingredients and parameter scales for a future implementation. Motional Fock-state branches with different kinetic energies induce different effective proper times~\cite{Sorci2026,Brewer2019}:
\begin{equation}
  \theta
  \simeq
  \frac{\omega_0 t\hbar\Omega}{4mc^2}\Delta n,
  \qquad
  \Delta n=|n_1-n_2|.
  \label{eq:theta_ion}
\end{equation}
Equivalently,
\begin{equation}
  \theta
  =
  \kappa
  \left(
    \frac{\Omega/2\pi}{1\,{\rm MHz}}
  \right)
  \left(
    \frac{t}{1\,{\rm s}}
  \right)
  \Delta n,
  \label{eq:theta_kappa}
\end{equation}
with representative values given in Table~\ref{tab:ion_params}.

\begin{table}[b]
  \centering
  \caption{Representative ion clock parameters (transition frequencies from Refs.~\cite{Ludlow2015,Brewer2019}).}
  \label{tab:ion_params}
  \begin{tabular}{lcc}
    \hline
    Ion clock & Frequency & $\kappa$\\
    \hline
    ${}^{27}{\rm Al}^+$          & $1.12\;{\rm PHz}$  & $2.9\times10^{-4}$\\
    ${}^{171}{\rm Yb}^+$ (E3)    & $642\;{\rm THz}$   & $2.6\times10^{-5}$\\
    ${}^{171}{\rm Yb}^+$ (E2)    & $688\;{\rm THz}$   & $2.8\times10^{-5}$\\
    ${}^{88}{\rm Sr}^+$          & $445\;{\rm THz}$   & $3.5\times10^{-5}$\\
    \hline
  \end{tabular}
\end{table}

The Al$^+$ platform is favorable because $\theta\propto\omega_0/m$. As an illustrative single-round scale,
\begin{equation}
  \begin{gathered}
    \Omega/2\pi=1\,{\rm MHz},
    \quad
    t=1\,{\rm s},
    \quad
    \Delta n=100,\\
    \theta\simeq 2.9\times10^{-2}.
  \end{gathered}
  \label{eq:illustrative_scale}
\end{equation}

A minimal selective-erasure protocol, based on standard trapped-ion motional-state preparation and readout primitives~\cite{Meekhof1996,Johnson2017,FluehmannHome2020}, prepares $|\chi_+\rangle=(|\chi_1\rangle+|\chi_2\rangle)/\sqrt2 = A|0\rangle$, lets the clock undergo the Ramsey histories $V_\pm=U_\pm R\,U_\pm$, and then applies $A^\dagger$ followed by motional ground-state or Fock-state readout. Successful postselection implements the projection onto the history-erasing state $|\chi_+\rangle$ and yields $K_+=(V_++V_-)/2$.

A phase-amplified implementation can be summarized by an effective parameter $\theta_{\rm eff}$, engineered so that the branch-dependent phase accumulates coherently. In a simple phenomenological visibility model, imperfect history-erasure visibility, clock coherence, and motional coherence are collected into a single parameter $\eta_{\rm tot}$, giving
\begin{equation}
  \langle Z\rangle_+
  =
  -\frac{\eta_{\rm tot}\sin^2\theta_{\rm eff}}
  {1+\eta_{\rm tot}\cos^2\theta_{\rm eff}}.
  \label{eq:visibility_model}
\end{equation}
This reduces to the ideal witness at $\eta_{\rm tot}=1$ and to the free-history prediction $\langle Z\rangle_+=0$ at $\eta_{\rm tot}=0$. The dark port trades this small bright-port imbalance for unit conditional contrast at success probability $p_-\simeq\theta_{\rm eff}^2/2$.

For numerical illustration, $\theta_{\rm eff}=0.1$ gives $|\langle Z\rangle_+|\simeq5.0\times10^{-3}$, while $\theta_{\rm eff}=0.3$ gives $|\langle Z\rangle_+|\simeq4.6\times10^{-2}$. A phase-amplified implementation with an effective accumulation factor $L_{\rm eff}\sim10$, for which $\theta_{\rm eff}\simeq L_{\rm eff}\theta$, would reach $\theta_{\rm eff}\simeq0.3$ and hence a percent-level witness. These estimates serve as feasibility-scale targets for future optical-clock implementations. Detailed visibility requirements and decoherence channels are discussed in Appendix~\ref{app:visibility}.

\section{Relation to prior work}
\label{sec:prior}

The present work connects four lines of research.

\emph{Quantum time dilation and clock proper time.} It is now established that quantum clocks carry relativistic phases, that motional time-dilation effects imprint on optical ion clocks, and that gravitational time dilation has been proposed to induce universal decoherence~\cite{Zych2011,Pikovski2015,SmithAhmadi2020,Sorci2026}. The complementary question --- which reduced or conditioned signals from such clocks are classically simulable and which are not --- is the focus of the present paper. We provide a certification boundary rather than a relativistic phase calculation.

\emph{Quantum eraser and self-interfering clocks.} Quantum eraser experiments restore interference by erasing which-path information~\cite{ScullyDruhl1982}, and self-interfering clocks have been used as ``which-path'' witnesses~\cite{Margalit2015}. Here erasure is used differently: the erased label must be a proper-time-history label, and the conditioned operation must leave the classical history-mixture set $\mathsf{CPTH}(\mathcal H)$. Ordinary erasure alone does not constitute a proper-time-history witness.

\emph{Quantum reference frames and relational time.} The Page--Wootters and related relational-time frameworks treat clocks or reference frames as quantum systems, describing dynamics without external time parameters~\cite{PageWootters1983,Giacomini2019}. Our approach is complementary: we take a clock as a probe of proper-time histories and ask when the resulting clock channel certifies nonclassical proper-time structure. These are different operational questions.

\emph{Resource theories and channel witnesses.} Resource theories of magic, coherence, and general quantum resources classify free (classically simulable) sets and quantify resource content via robustness and witness measures~\cite{HowardCampbell2017,Baumgratz2014,ChitambarGour2019}. We adopt the convex-operational logic of these approaches --- defining a free set $\mathsf{CPTH}(\mathcal H)$ and a witness of non-membership --- without claiming a full resource theory of quantum proper time. The free set is finite-dimensional and history-specific; the witness is the linear functional $\operatorname{Tr}[Z\,\Phi(|g\rangle\langle g|)]$ that vanishes on the free set but detects the conditioned operation. An optional robustness-type diagnostic is discussed in Appendix~\ref{app:robustness}.

\section{Discussion and outlook}
\label{sec:discussion}

We have distinguished three levels of proper-time signatures in quantum clocks. First, reduced clock dephasing, even when generated by quantum motion, can be written as a classical characteristic function. Second, averaged Ramsey histories form classical mixtures of the specified history unitaries. Third, conditioned Ramsey-history recombination can exceed such mixtures by producing coherent cross terms between histories. The usefulness of the construction comes from applying standard convex-operational tools to isolate a classical-simulable sector for specified proper-time histories.

The protocol uses established history-erasing and clock-interferometric ingredients to test a new simulability boundary. Selective erasure becomes a proper-time-history witness when the erased label determines distinct history unitaries $V_h$ and the conditioned operation falls outside $\mathsf{CPTH}(\mathcal H)$. The same convex set supports a robustness-type quantifier and linear witness bounds (Appendix~\ref{app:robustness}), extending the resource-theoretic idea of quantifying distance from a classically simulable sector.

The optical-clock mapping identifies a route toward a minimal certification protocol. Motional branches with different kinetic energies label different effective proper-time histories; inverse motional preparation and ground-state or Fock-state readout implement the history-erasing projection. The feasibility estimates set a target scale for future implementations.

To summarize: nonclassical proper-time histories are certified by the failure of classical mixtures of the same specified histories, not by reduced dephasing alone. We have given a minimal sufficient witness and certification framework. A complete resource theory or necessary-and-sufficient characterization of all forms of quantum proper-time nonclassicality remains open, as do extensions to multi-history sets, continuous history labels, and gravitational proper-time scenarios.

\begin{acknowledgments}
  This work was supported by the CPS-Yangtze Delta Region Industrial Innovation Center of Quantum and Information Technology-MindSpore Quantum Open Fund.
\end{acknowledgments}

\bibliographystyle{quantum}
\bibliography{proper_time_prl}

\onecolumn
\appendix

\section{Two-level dephasing no-go}
\label{app:dephasing}

We provide the explicit construction referenced in Proposition~\ref{prop:dephasing}. Consider a two-level clock with transition frequency $\omega$. A single-time dephasing signal is specified by a complex number $\Gamma$ satisfying $|\Gamma|\le1$. Write $\Gamma=re^{i\phi}$ with $0\le r\le1$. Choose $\tau_1,\tau_2$ such that
\begin{equation}
  e^{-i\omega\tau_1}=e^{i\phi},
  \qquad
  e^{-i\omega\tau_2}=-e^{i\phi}.
\end{equation}
The probability measure
\begin{equation}
  d\mu(\tau)
  =
  \frac{1+r}{2}\delta(\tau-\tau_1)d\tau
  +
  \frac{1-r}{2}\delta(\tau-\tau_2)d\tau
\end{equation}
gives $\Gamma=\int d\mu(\tau)e^{-i\omega\tau}$. Therefore every two-level single-time dephasing factor is simulable by a classical random proper-time distribution. The result is tight: the unit disk is exactly the convex hull of the pure-phase circle $|e^{-i\omega\tau}|=1$, so the dephasing parameter space saturates the classical-simulable region.

\section{Effective proper-time label}
\label{app:effective_label}

We expand the construction of Proposition~\ref{prop:effective}. Let $H_C=\sum_n E_n |n\rangle\langle n|$ and consider the effective joint evolution
\begin{equation}
  U=e^{-iH_C\otimes\hat\tau}.
\end{equation}
Here $\hat\tau$ is used as a self-adjoint label of motional proper-time branches, not as a universal time operator conjugate to the Hamiltonian. For an initial state $\rho_C\otimes\rho_M$, the clock coherence $|m\rangle\langle n|$ is multiplied by
\begin{equation}
  \Gamma_{mn}
  =
  \operatorname{Tr}_M
  \left[
    \rho_M e^{-i(E_m-E_n)\hat\tau}
  \right].
\end{equation}
Let $P(d\tau)$ denote the spectral measure of $\hat\tau$. Then
\begin{equation}
  \Gamma_{mn}
  =
  \int e^{-i(E_m-E_n)\tau}
  \operatorname{Tr}_M[\rho_M P(d\tau)].
\end{equation}
The measure $d\mu_{\rho_M}(\tau)=\operatorname{Tr}_M[\rho_M P(d\tau)]$ is a probability measure. Hence the reduced clock signal is equivalent to a classical random proper-time distribution, even though the label arises from a quantum system.

A concrete branch-label model can be written for Fock-state motional branches. If the branch $|n\rangle$ induces
\begin{equation}
  \tau_n \simeq t-\frac{t}{2m^2c^2}\langle p^2\rangle_n,
  \qquad
  \langle p^2\rangle_n=m\hbar\Omega\left(n+\frac12\right),
\end{equation}
then one may take
\begin{equation}
  \hat\tau=\sum_n \tau_n |n\rangle\langle n|,
  \qquad
  P(d\tau)=\sum_n |n\rangle\langle n|\,\delta(\tau-\tau_n)d\tau .
\end{equation}
In this concrete model, $\hat\tau$ functions as a proper-time branch label tied to the motional spectrum.

\section{Choi-rank proof and membership test}
\label{app:choi}

We provide the full Choi-rank argument for Theorem~\ref{thm:choi}. Assume a conditioned quantum operation has a single Kraus operator
\begin{equation}
  K=\sum_{h\in\mathcal H} c_h V_h
\end{equation}
and satisfies $K^\dagger K=pI$ with $p>0$. Then $W=K/\sqrt p$ is unitary and the conditioned operation is $\mathcal I(\rho)=pW\rho W^\dagger$.

Suppose for contradiction that $\mathcal I\in\mathsf{CPTH}_{\rm inst}(\mathcal H)$. Then
\begin{equation}
  pW\rho W^\dagger
  =
  \sum_{h\in\mathcal H} q_h V_h\rho V_h^\dagger,
  \qquad q_h\ge0.
\end{equation}
Taking Choi representations gives
\begin{equation}
  p|W\rangle\rangle\langle\langle W|
  =
  \sum_{h\in\mathcal H} q_h|V_h\rangle\rangle\langle\langle V_h|.
\end{equation}
The left-hand side has rank one. A nonnegative sum of rank-one projectors has rank one only if all nonzero vectors in the sum are parallel. Hence every $h$ with $q_h>0$ satisfies $V_h\propto W$. Therefore, if $W$ is not proportional to any $V_h$ in the specified history set $\mathcal H$, no classical mixture of those histories can reproduce the conditioned operation.

For outcomes that do not satisfy $K^\dagger K\propto I$ or that have multiple Kraus operators, one must use the full Choi cone membership test of Lemma~\ref{lem:choi_feasibility}. The test is a finite-dimensional semidefinite feasibility problem: checking whether a given Choi operator $J(\mathcal I)$ can be written as a nonnegative linear combination of the known $J(\mathcal U_h)$ for $h\in\mathcal H$.

\section{Ramsey algebra}
\label{app:ramsey}

We provide the explicit calculation for the minimal two-history Ramsey witness. Let
\begin{equation}
  U_\pm=e^{\mp i\theta Z/2},
  \qquad
  R=\frac{1}{\sqrt2}
  \begin{pmatrix}
    1&1\\1&-1
  \end{pmatrix},
\end{equation}
and define $V_+=U_+R\,U_+$ and $V_-=U_-R\,U_-$.

\paragraph{Bright port.}
For the specified two-history set $\mathcal H_2=\{+,-\}$, the history-erased Kraus operator is
\begin{equation}
  K_+=\frac12(V_++V_-)
  =
  \frac{1}{\sqrt2}
  \begin{pmatrix}
    \cos\theta & 1\\
    1 & -\cos\theta
  \end{pmatrix}.
\end{equation}
Thus
\begin{equation}
  K_+^\dagger K_+
  =
  \frac{1+\cos^2\theta}{2}I,
  \qquad
  p_+=\frac{1+\cos^2\theta}{2}.
\end{equation}
For input $|g\rangle=(1,0)^{\mathsf T}$,
\begin{equation}
  K_+|g\rangle
  =
  \frac{1}{\sqrt2}(\cos\theta|g\rangle+|e\rangle),
\end{equation}
and the normalized output is
\begin{equation}
  |\psi_+\rangle
  =
  \frac{\cos\theta|g\rangle+|e\rangle}{\sqrt{1+\cos^2\theta}}.
\end{equation}
Therefore,
\begin{equation}
  \langle Z\rangle_+
  =
  \frac{\cos^2\theta-1}{\cos^2\theta+1}
  =
  -\frac{\sin^2\theta}{1+\cos^2\theta}.
\end{equation}
For every classical mixture in $\mathsf{CPTH}(\mathcal H_2)$, each history maps $|g\rangle$ to an equal-population superposition, and hence $\langle Z\rangle_{\rm cl}=0$.

\paragraph{Dark port.}
The complementary outcome is
\begin{equation}
  K_-=\frac12(V_+-V_-)
  =
  -\frac{i\sin\theta}{\sqrt2}I,
\end{equation}
proportional to the identity. Then
\begin{equation}
  K_-^\dagger K_-=\frac{\sin^2\theta}{2}I,
  \qquad
  p_-=\frac{\sin^2\theta}{2}.
\end{equation}
Since $W_-=I$ up to a global phase, we have $\langle Z\rangle_- = \langle Z\rangle_{\rm in}$ for any input. For input $|g\rangle$ with $Z|g\rangle=|g\rangle$, this gives unit conditional contrast $\langle Z\rangle_-=1$. The dark-port success probability scales as $\theta^2/2$ for small $\theta$, so the number of experimental trials required to reach $s\sigma$ significance of successful dark-port events is $N_{\rm trial}^{(-)}\sim 2s^2/\theta^2$ (or $2s^2/\theta_{\rm eff}^2$ for a phase-amplified effective angle).

\section{Robustness-type diagnostics}
\label{app:robustness}

The free set of normalized classical proper-time-history channels is $\mathsf{CPTH}(\mathcal H)$. This diagnostic is inspired by the general resource-theoretic logic used in magic-state and coherence resource theories~\cite{HowardCampbell2017,Baumgratz2014,ChitambarGour2019}. We use the term ``robustness'' only as a diagnostic relative to the fixed finite history set, not as a complete resource theory. For a normalized postselected channel $\Phi$, one may define the robustness-type witness diagnostic
\begin{equation}
  R_{\rm PTH}^{\mathcal H}(\Phi)
  =
  \inf_{\Psi}
  \left\{
    s\ge0:
    \frac{\Phi+s\Psi}{1+s}\in\mathsf{CPTH}(\mathcal H)
  \right\},
  \label{eq:robustness}
\end{equation}
where the infimum is over quantum channels $\Psi$ acting on the clock. The value may be infinite for channels outside the affine span reachable by this mixing. This quantity optionally quantifies the obstruction to simulation by the specified history set; the separation theorem itself does not rely on it. A complete resource theory of quantum proper time is left for future work.

By construction,
\begin{equation}
  R_{\rm PTH}^{\mathcal H}(\Phi)=0
  \quad\Longleftrightarrow\quad
  \Phi\in\mathsf{CPTH}(\mathcal H).
\end{equation}
If $\Phi=\mathcal U_W$ is a unitary channel and $W$ is not proportional to any $V_h$ with $h\in\mathcal H$, then the Choi-rank theorem implies $R_{\rm PTH}^{\mathcal H}(\Phi)>0$.

A linear witness gives a lower bound. Let $F$ be a real linear functional on channels such that $F(\Phi_{\rm cl})=0$ for all $\Phi_{\rm cl}\in\mathsf{CPTH}(\mathcal H)$ and $|F(\Psi)|\le1$ for all quantum channels $\Psi$. If $(\Phi+s\Psi)/(1+s)\in\mathsf{CPTH}(\mathcal H)$, then $s\ge |F(\Phi)|$, and therefore
\begin{equation}
  R_{\rm PTH}^{\mathcal H}(\Phi)\ge |F(\Phi)|.
\end{equation}
For the minimal Ramsey protocol,
\begin{equation}
  F(\Phi)=\operatorname{Tr}\left[Z\,\Phi(|g\rangle\langle g|)\right]
\end{equation}
vanishes on $\mathsf{CPTH}(\mathcal H_2)$, whereas $F(\Phi_+)=\langle Z\rangle_+$. Hence
\begin{equation}
  R_{\rm PTH}^{\mathcal H_2}(\Phi_+)
  \ge
  \frac{\sin^2\theta}{1+\cos^2\theta}.
\end{equation}

The linear witness and robustness-type diagnostic therefore supply optional resource-theoretic quantifiers of the same obstruction detected by the Choi-rank criterion.

\section{Visibility and trapped-ion estimates}
\label{app:visibility}

\subsection{Phenomenological visibility model}

For imperfect coherence, we use a phenomenological visibility model. Let
\begin{equation}
  \eta_{\rm tot}=\eta_h\eta_c\eta_m
\end{equation}
combine history-erasure visibility, clock coherence, and motional-branch coherence. In this simple model,
\begin{equation}
  \langle Z\rangle_+
  =
  -\frac{\eta_{\rm tot}\sin^2\theta_{\rm eff}}
  {1+\eta_{\rm tot}\cos^2\theta_{\rm eff}}.
\end{equation}
This expression reduces to the ideal witness at $\eta_{\rm tot}=1$ and to the free-history prediction $\langle Z\rangle_+=0$ at $\eta_{\rm tot}=0$.

The factors have direct physical interpretations. The factor $\eta_h$ summarizes the fidelity of the selective history-erasure operation, including imperfect inverse motional preparation, residual distinguishability of the motional branches, and readout errors in the postselection. The factor $\eta_c$ describes ordinary clock coherence during the Ramsey sequence, including laser phase noise and clock-state dephasing. The factor $\eta_m$ describes the coherence of the motional branch superposition; motional heating, trap-frequency noise, and off-resonant coupling reduce this factor.

This visibility decomposition identifies which experimental quantities must remain appreciable for the witness to survive. A dedicated implementation would require platform-specific optimization of motional-state preparation, heating rate, clock interrogation time, and postselection fidelity.

\subsection{Effective amplification}

The single-round witness scales as $\langle Z\rangle_+\simeq -\theta^2/2$ for $\theta\ll1$. For feasibility estimates we introduce an effective accumulated phase $\theta_{\rm eff}$ for an engineered phase-amplified implementation. In a sequence designed to accumulate the branch-dependent phase coherently, one may have approximately $\theta_{\rm eff}\simeq L\theta$, but this scaling is protocol-dependent.

With the replacement $\theta\mapsto\theta_{\rm eff}$, the ideal witness becomes
\begin{equation}
  \langle Z\rangle_+
  =
  -\frac{\sin^2\theta_{\rm eff}}{1+\cos^2\theta_{\rm eff}},
\end{equation}
and, for small $\theta_{\rm eff}$,
\begin{equation}
  \langle Z\rangle_+\simeq -\frac{\theta_{\rm eff}^2}{2}.
\end{equation}
The successful postselected shot estimate at $s\sigma$ significance is
\begin{equation}
  N_{\rm succ}\sim \frac{s^2}{\langle Z\rangle_+^2}
  \sim \frac{4s^2}{\theta_{\rm eff}^4}
\end{equation}
in the small-angle ideal limit.

\subsection{Trapped-ion parameters}

For Fock-state motional branches,
\begin{equation}
  \theta
  \simeq
  \frac{\omega_0t\hbar\Omega}{4mc^2}\Delta n.
\end{equation}
Representative values, based on standard optical-ion-clock transition scales and trapped-ion motional parameters~\cite{Sorci2026,Brewer2019}, are shown in Table~\ref{tab:ion_params}. They are included to set an illustrative scale; the visibility factors discussed above determine whether such a scale can be reached in a platform-specific implementation.

The Al$^+$ platform is favorable because $\theta\propto\omega_0/m$. As an illustrative single-round scale,
\begin{equation}
  \begin{gathered}
    \Omega/2\pi=1\,{\rm MHz},
    \qquad
    t=1\,{\rm s},
    \qquad
    \Delta n=100,\\
    \theta\simeq 2.9\times10^{-2}.
  \end{gathered}
\end{equation}
A phase-amplified implementation with an effective accumulation factor $L_{\rm eff}\sim10$, for which $\theta_{\rm eff}\simeq L_{\rm eff}\theta$, would reach $\theta_{\rm eff}\simeq0.3$ and hence a percent-level witness.

\end{document}